\documentclass[11pt,USletter]{article}
\usepackage{fullpage}

\usepackage[utf8]{inputenc}
\usepackage{amsthm,amsmath,amssymb}
\usepackage{thmtools}
\usepackage{subcaption,thm-restate}
\usepackage[mathcal,mathscr]{eucal}
\usepackage{epsfig,graphicx,graphics,color}
\usepackage{caption}
\usepackage[position=b]{subcaption}
\usepackage[dvipsnames]{xcolor} 
\usepackage{enumerate}
\usepackage[sort,nocompress]{cite}
\usepackage{hyperref}
\hypersetup{
    colorlinks=true,       
    linkcolor=blue,        
    citecolor=red,         
    filecolor=magenta,     
    urlcolor=cyan,         
    linktocpage=true
}

\theoremstyle{plain}
\newtheorem{thm}{Theorem}[section]
\newtheorem{lem}[thm]{Lemma}
\newtheorem{cor}[thm]{Corollary}
\newtheorem{cl}[thm]{Claim}

\newtheorem{conj}[thm]{Conjecture}

\theoremstyle{definition}

\newtheorem{rem}[thm]{Remark}

\def\final{0}  
\def\iflong{\iffalse}
\ifnum\final=0  
\usepackage[dvipsnames]{xcolor}
\newcommand{\knote}[1]{{\color{red}[{\tiny \textbf{Kristóf:} \bf #1}]\marginpar{\color{red}*}}}
\newcommand{\ynote}[1]{{\color{blue}[{\tiny \textbf{Yutaro:} \bf #1}]\marginpar{\color{blue}*}}}
\newcommand{\tnote}[1]{{\color{OliveGreen}[{\tiny \textbf{Tamás:} \bf #1}]\marginpar{\color{OliveGreen}*}}}
\else 
\newcommand{\knote}[1]{}
\newcommand{\ynote}[1]{}
\newcommand{\tnote}[1]{}
\fi

\DeclareMathOperator\dist{dist}

\newcommand{\cB}{\mathcal{B}}
\newcommand{\cI}{\mathcal{I}}
\newcommand{\cJ}{\mathcal{J}}
\newcommand{\cC}{\mathcal{C}}

\DeclareRobustCommand{\rchi}{{\mathpalette\irchi\relax}}
\newcommand{\irchi}[2]{\raisebox{\depth}{$#1\chi$}}
\newcommand{\overbar}[1]{\mkern 3mu\overline{\mkern-3mu#1\mkern-3mu}\mkern 3mu}
\makeatletter
\newcommand{\dbloverline}[1]{\overbar{\dbl@overline{#1}}}
\newcommand{\dbl@overline}[1]{\mathpalette\dbl@@overline{#1}}
\newcommand{\dbl@@overline}[2]{%
  \begingroup
  \sbox\z@{$\m@th#1\overbar{#2}$}%
  \ht\z@=\dimexpr\ht\z@-1.5\dbl@adjust{#1}\relax
  \box\z@
  \ifx#1\scriptstyle\kern-\scriptspace\else
  \ifx#1\scriptscriptstyle\kern-\scriptspace\fi\fi
  \endgroup
}
\newcommand{\dbl@adjust}[1]{%
  \fontdimen8
  \ifx#1\displaystyle\textfont\else
  \ifx#1\textstyle\textfont\else
  \ifx#1\scriptstyle\scriptfont\else
  \scriptscriptfont\fi\fi\fi 3
}
\makeatother
\newcommand\wtm{R}
\newcommand\wtmm{\overbar{R}}

\title{List coloring of two matroids\\ through reduction to partition matroids}

\author{
Kristóf Bérczi\thanks{MTA-ELTE Egerváry Research Group, Department of Operations Research, Eötvös Loránd University, Budapest, Hungary. Email: \texttt{berkri@cs.elte.hu}.}
\and
Tamás Schwarcz\thanks{Department of Operations Research, Eötvös Loránd University, Budapest, Hungary. Email: \texttt{schtomi97@gmail.com}.}
\and
Yutaro Yamaguchi\thanks{Osaka University, Osaka, Japan. Email: \texttt{yutaro\_yamaguchi@ist.osaka-u.ac.jp}.}
}

\begin{document}
\maketitle

\begin{abstract}
In the list coloring problem for two matroids, we are given matroids $M_1=(S,\cI_1)$ and $M_2=(S,\cI_2)$ on the same ground set $S$, and the goal is to determine the smallest number $k$ such that given arbitrary lists $L_s$ of $k$ colors for $s\in S$, it is possible to choose a color from each list so that every monochromatic set is independent in both $M_1$ and $M_2$. When both $M_1$ and $M_2$ are partition matroids, Galvin's celebrated list coloring theorem for bipartite graphs gives the answer. However, not much is known about the general case. One of the main open questions is to decide if there exists a constant $c$ such that if the coloring number is $k$ (i.e., the ground set can be partitioned into $k$ common independent sets), then the list coloring number is at most $c\cdot k$. In the present paper, we consider matroid classes that appear naturally in combinatorial and graph optimization problems, namely graphic matroids, paving matroids and gammoids. We show that if both matroids are from these fundamental classes, then the list coloring number is at most twice the coloring number. 

The proof is based on a new approach that reduces a matroid to a partition matroid without increasing its coloring number too much, and might be of independent combinatorial interest. In particular, we show that if $M=(S,\cI)$ is a matroid in which $S$ can be partitioned into $k$ independent sets, then there exists a partition matroid $N=(S,\cJ)$ with $\cJ\subseteq\cI$ in which $S$ can be partitioned into (A) $k$ independent sets if $M$ is a transversal matroid, (B) $2k-1$ independent sets if $M$ is a graphic matroid, (C) $\lceil kr/(r-1)\rceil$ independent sets if $M$ is a paving matroid of rank $r$, and (D) $2k-2$ independent sets if $M$ is a gammoid. It should be emphasized that in cases (A), (B) and (D) the rank of $N$ is the same as that of $M$. We extend our results to a broader family of matroids by showing that the existence of a matroid $N$ with $\rchi(N)\leq 2\rchi(M)$ implies the existence of a matroid $N'$ with $\rchi(N')\leq 2\rchi(M')$ for every truncation $M'$ of $M$. We also show how the reduction technique can be extended to strongly base orderable matroids that might serve as a useful tool in problems related to packing common bases of two matroids.

\medskip

\noindent \textbf{Keywords:} List coloring, Matroids, Packing common bases, Strongly base orderable matroids 
\end{abstract}

\section{Introduction} \label{sec:intro}

Given a graph $G=(V,E)$, a \textbf{proper edge coloring} of $G$ is an assignment of colors to the edges so that no two adjacent edges have the same color. The \textbf{edge coloring number} is the smallest integer $k$ for which $G$ has a proper edge coloring by $k$ colors. The classical result of K\H{o}nig \cite{konig1916graphen} states that the edge coloring number of bipartite graphs is equal to its maximum degree.

Assume now that a list $L_e$ of colors is given for each edge $e\in E$. A \textbf{proper list edge coloring} of $G$ is a proper edge coloring such that every edge $e$ receives a color from its list $L_e$. The \textbf{list edge coloring number} is the smallest integer $k$ for which $G$ has a proper list edge coloring whenever $|L_e|\geq k$ for every $e\in E$. The List Coloring Conjecture \cite{jensen2011graph,vizing1965chromatic} states that for any graph, the list edge coloring number equals the edge coloring number. The conjecture is widely open, and only partial results are known. The probably most famous one is the celebrated result of Galvin \cite{galvin1995list} who showed that the conjecture holds for bipartite multigraphs. 

\begin{thm}[Galvin]\label{thm:galvin}
The list edge coloring number of a bipartite graph is equal to its edge coloring number, that is, to its maximum degree. 
\end{thm}

Matchings in bipartite graphs are forming the common independent sets of two matroids, hence one might consider matroidal generalizations of list coloring. For a loopless matroid\footnote{The ground set of a matroid containing a loop cannot be decomposed into independent sets. Therefore every matroid considered in the paper is assumed to be loopless without explicitly mentioning this. Nevertheless, parallel elements might exist.} $M=(S,r)$, let $\rchi(M)$ denote the \textbf{coloring number} of $M$, that is, the minimum number of independent sets into which the ground set can be decomposed in $M$. We call a matroid \textbf{$k$-colorable} if $\rchi(M)\leq k$. If a list $L_s$ of colors is given for each element $s\in S$, then a \textbf{proper list edge coloring} of $M$ is a coloring such that every element $s$ receives a color from its list $L_s$. The \textbf{list coloring number} is the smallest integer $k$ for which $M$ has a proper list coloring whenever $|L_s|\geq k$ for every $s\in S$. 

Analogously, the \textbf{coloring number} $\rchi(M_1,M_2)$ of the intersection of two matroids $M_1$ and $M_2$ on the same ground set $S$ is the minimum number of common independent sets needed to cover $S$. If a list $L_s$ of colors is given for each element $s\in S$, then a \textbf{proper list edge coloring} of the intersection of $M_1$ and $M_2$ is a coloring such that every element $s$ receives a color from its list $L_s$. The \textbf{list coloring number} $\rchi_\ell(M_1,M_2)$ is the smallest integer $k$ for which the intersection of $M_1$ and $M_2$ has a proper list coloring whenever $|L_s|\geq k$ for every $s\in S$. Hence Theorem~\ref{thm:galvin} states that if both $M_1$ and $M_2$ are partition matroids then $\rchi_\ell(M_1,M_2)=\max\{\rchi(M_1),\rchi(M_2)\}$.

\paragraph{Previous work}
Seymour observed \cite{seymour1998note} that the list coloring theorem holds for a single matroid.

\begin{thm}[Seymour] \label{thm:seymour}
The list coloring number of a matroid is equal to its coloring number.
\end{thm}

Laso{\'n} \cite{lason2015list} gave a generalization of the theorem when the sizes of the lists are not necessarily equal. As a common generalization of Theorems~\ref{thm:galvin} and \ref{thm:seymour}, it is tempting to conjecture that $\rchi(M_1,M_2)=\rchi_\ell(M_1,M_2)$ holds for every pair of matroids \cite{kiraly2013open}. 
No pair $M_1,M_2$ is known for which the conjecture fails. Nevertheless, there are only a few matroid classes for which the problem was settled. Kir\'aly and Pap \cite{kiraly2010list} verified the conjecture for transversal matroids, for matroids of rank two, and if the common bases are the arborescences of a digraph which is the disjoint union of two spanning arborescences rooted at the same vertex. It is worth mentioning that a similar statement does not hold for the case of three matroids as shown by the following example due to Kir\'aly \cite{egresopen}. Let $S=\{a,b,c,d,e,f\}$ be a ground set of size six, and let $M_1$, $M_2$ and $M_3$ be partition matroids with circuit sets $\mathcal{C}_1=\{\{a,d\},\{b,e\},\{c,f\}\}$, $\mathcal{C}_2=\{\{a,e\},\{b,f\},\{c,d\}\}$ and $\mathcal{C}_3=\{\{a,f\},\{b,d\},\{c,e\}\}$, respectively. Then $\{a,b,c\}$ and $\{d,e,f\}$ is a partition into two common bases. However, if $L_a=L_d=\{1,2\}$, $L_b=L_e=\{1,3\}$ and $L_c=L_f=\{2,3\}$, then there is no proper list coloring. 
In \cite{egresopen}, Kir\'aly proposed a weakening of the problem where the aim is to find a constant $c$ such that if the coloring number is $k$, then the list coloring number is at most $c\cdot k$. For spanning arborescences, it was observed by Kobayashi \cite{egresopen} that the constructive characterization of $k$-arborescences implies that lists of size $\frac{3}{2}k+1$ are sufficient.

Another motivation comes from the problem of approximating the minimum number of common independent sets of two matroids needed to cover the ground set. Aharoni and Berger \cite{aharoni2006intersection} proved the following interesting result. 

\begin{thm}[Aharoni and Berger] \label{thm:ab}
Let $M_1=(S,\cI_1)$ and $M_2=(S,\cI_2)$ be matroids. If $S$ can be decomposed into $k_1$ independent sets in $M_1$ and into $k_2$ independent sets in $M_2$, then it can be decomposed into $2\max\{k_1,k_2\}$ common independent sets.
\end{thm}

Note that in terms of approximating the coloring number, Theorem~\ref{thm:ab} states that $\rchi(M_1,M_2)\leq 2\max\{\rchi(M_1),\rchi(M_2)\}$. The proof of the theorem is highly non-trivial and uses topological arguments which do not directly give an algorithm for finding the decomposition in question. We propose a conjecture on reduction of matroids to partition matroids from which the theorem would easily follow.

The idea of reducing a matroid to a simpler one goes back to the late 60's. In \cite{crapo1968combinatorial}, Crapo and Rota introduced the notion of weak maps. Given two matroids $M$ and $N$ on the same ground set, $N$ is a \textbf{weak map} of $M$ if every independent set of $N$ is also independent in $M$. Using our terminology, $N$ is a weak map of $M$ if and only if $N$ is a reduction of $M$. Weak maps were further investigated by Lucas \cite{lucas1974properties,lucas1975weak} who characterized rank preserving weak maps for linear matroids. However, these results did not consider the possible increase in the coloring number of the matroid that plays a crucial role in our investigations. We find the name `map' slightly misleading as it suggests that there is a function in the background, although the `mapping' in question is simply the identity map between the ground sets of the matroids. Hence we stick to the term `reduction' throughout the paper. 

It is worth mentioning that every matroid $M=(S,\cI)$ has a reduction to a partition matroid $N=(S,\cJ)$ of the same rank. The sketch of the proof is as follows: Fix an arbitrary basis $B=\{s_1,\dots,s_r\}$ of $M$, and add $s_i$ to the $i$th partition class. Then for an arbitrary element $s\in S-B$, consider the fundamental circuit $C(s,B)$ of $s$ with respect to $B$, and add $s$ to the partition class containing the element of $C(s,B)\cap B$ with the smallest index. If we pick exactly one element from every class of the partition thus obtained, we get a basis of the matroid. This can be verified using the circuit axioms, not discussed in this paper. Nevertheless, this algorithm has no control over the sizes of the partition classes. It can happen that some of the classes have a large size compared to the coloring number of the original matroid, and such a reduction is not suitable for our purposes.

\paragraph{Our results} 
In the present paper, we consider matroid classes that appear naturally in combinatorial and graph optimization problems, and show that if both matroids are from these fundamental classes then $c$ can be chosen to be roughly $2$. Our proof builds on the reduction of a matroid to a partition matroid. Given matroids $M=(S,\cI)$ and $N=(S,\cJ)$, we say that $N$ is a \textbf{reduction} of $M$ if $\cJ\subseteq\cI$, that is, every independent set of $N$ is independent in $M$ as well. In notation, we will denote $N$ being a reduction of $M$ by $N\preceq M$. The reduction is \textbf{rank preserving} if $r_M(S)=r_N(S)$ holds, and is denoted by $N\preceq_r M$.

A \textbf{partition matroid}\footnote{In the literature, partition matroids are defined more generally in the sense that the upper bounds on the intersection might be different for the different partition classes. As all the partition matroids used in the paper have all-ones upper bounds, we make this restriction without explicitly mentioning it.} is a matroid $N=(S,\cJ)$ such that $\cJ=\{X\subseteq S:\ |X\cap S_i|\leq 1\ \text{for $i=1,\dots,q$}\}$ for some partition $S=S_1\cup\dots\cup S_q$. Clearly, the coloring number of $N$ is $\rchi(N)=\max\{|S_i|:\ i=1,\dots,q\}$. Note that the followings are equivalent for a matroid $M=(S,\cI)$: (i) $N\preceq M$, (ii) every circuit of $M$ intersects at least one of the $S_i$'s in more than one element, and (iii) $\{x_1,\dots,x_q\}\in\cI$ whenever $x_i\in S_i$ for $i=1,\dots,q$.

To illustrate the applicability of reduction, assume that $M_1$ and $M_2$ are matroids on the same ground set that are reducible to $k$-colorable partition matroids $N_1$ and $N_2$, respectively. Then, by Theorem~\ref{thm:galvin}, $\rchi_\ell(N_1,N_2)\leq k$. As $N_1\preceq M_1$ and $N_2\preceq M_2$, this in turn implies that $\rchi_\ell(M_1,M_2)\leq k$. 

As a first step towards understanding reducibility to partition matroids, we concentrate on matroid classes that appear naturally in combinatorial and graph optimization problems, namely graphic matroids, paving matroids, and gammoids; see the corresponding sections for the precise definitions. These classes also include uniform matroids, laminar matroids and transversal matroids, hence the presented results also apply to those. We show that matroids from these fundamental classes admit a reduction to a partition matroid with coloring number at most twice the coloring number of the original matroid.

The first two result are for transversal and graphic matroids and are based on easy observations. Although transversal matroids are special cases of gammoids, we discuss them separately as transversal matroids admit an optimal reduction in terms of coloring number. 

\begin{restatable}{thm}{thmtransversal}
\label{thm:transversal}
Let $M=(S,\cI)$ be a $k$-colorable transversal matroid. Then there exists a $k$-colorable partition matroid $N$ with $N\preceq_r M$.
\end{restatable}

\begin{restatable}{thm}{thmgraphic}
\label{thm:graphic}
Let $M=(S,\cI)$ be a $k$-colorable graphic matroid. Then there exists a $(2k-1)$-colorable partition matroid $N$ with $N\preceq_r M$, and the bound for the coloring number of $N$ is tight.
\end{restatable}

The next result is for paving matroids.

\begin{restatable}{thm}{thmpaving}
\label{thm:paving}
Let $M=(S,\cI)$ be a $k$-colorable paving matroid of rank $r \ge 2$. Then there exists a $\lceil\frac{rk}{r-1}\rceil$-colorable partition matroid $N$ with $N\preceq M$.
\end{restatable}

For $r=2$, the bound on the coloring number of $N$ can be improved and the reduction can be chosen to be rank preserving. 

\begin{restatable}{thm}{thmpavingtwo}
\label{thm:paving2}
Let $M=(S,\cI)$ be a $k$-colorable paving matroid of rank $2$. Then there exists a $\lfloor\frac{4k}{3}\rfloor$-colorable partition matroid $N$ with $N\preceq_r M$, and the bound for the coloring number of $N$ is tight.
\end{restatable}

It is not difficult to see that every loopless matroid of rank $2$ is paving, hence Theorem~\ref{thm:paving2} gives a tight bound on the coloring number of the reduction $N$ of such matroids.

For $r=3$, we can provide rank preserving reductions at the price of increasing the coloring number of $N$.

\begin{restatable}{thm}{thmpavingthree}
\label{thm:paving3}
Let $M=(S,\cI)$ be a $k$-colorable paving matroid of rank $3$. Then there exists a $(2k-1)$-colorable partition matroid $N$ with $N\preceq_r M$.
\end{restatable}

The main contribution of the paper is a proof that any $k$-colorable gammoid can be reduced to a $(2k-2)$-colorable partition matroid for $k\geq 2$. The assumption $k\geq 2$ is not restrictive, as if $k=1$ then $M$ is the free matroid on $S$ which is already a partition matroid.

\begin{restatable}{thm}{thmgammoid}
\label{thm:3}\label{thm:gammoid}
Let $M=(S,\cI)$ be a $k$-colorable gammoid ($k\geq 2$). Then there exists a $(2k-2)$-colorable partition matroid $N$ with $N\preceq_r M$, and the bound for the coloring number of $N$ is tight.
\end{restatable}

The proof of Theorem~\ref{thm:gammoid} is based on building up an alternating structure on degree-bounded trees in a bipartite graph and is interesting on its own. We believe that this approach works in general, and a similar proof can be given to the following conjecture.

\begin{conj}\label{con:red}
Every $k$-colorable matroid can be reduced to a $2k$-colorable partition matroid.
\end{conj}

Recall that the problem proposed by Kir\'aly asks for the existence of a constant $c$ such that if the coloring number $\rchi(M_1,M_2)$ is $k$, then the list coloring number $\rchi_\ell(M_1,M_2)$ is at most $c\cdot k$. Conjecture~\ref{con:red} and Theorem~\ref{thm:galvin} would imply $\rchi_\ell(M_1,M_2)\leq\rchi_\ell(N_1,N_2)=\max\{\rchi(N_1),\rchi(N_2)\}\leq2\max\{\rchi(M_1),\rchi(M_2)\}$, where $N_1$ and $N_2$ are the reductions of $M_1$ and $M_2$, respectively, provided by the conjecture. Therefore we would get a value of $2$ for the constant $c$. By the same idea, Conjecture~\ref{con:red} and K\H{o}nig's classical theorem together would immediately provide a new, hopefully algorithmic proof of Theorem~\ref{thm:ab}. 

We extend our results to broader classes of matroids by showing the following connection between truncation and reducibility.

\begin{restatable}{thm}{thmtruncation}
\label{thm:truncation}
The family of matroids $M$ that can be reduced to a $2\rchi(M)$-colorable partition matroid is closed for truncation.
\end{restatable}

\paragraph{Organization}
The rest of the paper is organized as follows. Basic definitions and notation are introduced in Section~\ref{sec:preliminaries}. Results for the transversal, graphic and paving cases are presented in Section~\ref{sec:easy}. The main result of the paper, Theorem~\ref{thm:gammoid} is proved in Section~\ref{sec:gammoid}. The connection between truncation and reducibility is discussed in Section~\ref{sec:trun}. Finally, a more general framework together with some open problems are proposed in Section~\ref{sec:sbo}.  
 
\section{Preliminaries} \label{sec:preliminaries}

For a graph $G=(V,E)$ and a subset $X\subseteq V$ of vertices, the \textbf{set of edges spanned by $X$} is denoted by $E[X]$, while the \textbf{graph spanned by $X$} is denoted by $G[X]$. Given a connected component $K$ of $G$, a \textbf{cut} of $K$ is a subset of edges in $E[K]$ whose deletion disconnects $K$. The component is \textbf{$k$-edge-connected} if the minimum size of a cut in $K$ is at least $k$. The graphs obtained by deleting a subset $X\subseteq V$ of vertices or a subset $F\subseteq E$ of edges are denoted by $G-X$ and $G-F$, respectively. The \textbf{degree} of a vertex $v$ with respect to $F\subseteq E$ is denoted by $d_F(v)$. The \textbf{symmetric difference} of two sets $P,Q$ is denoted by $P\triangle Q = (P - Q) \cup (Q - P)$.

Let $G=(A,B;E)$ be a bipartite graph and $F\subseteq E$ be a subset of edges. For a set $X\subseteq A$, the \textbf{set of neighbors of $X$ with respect to $F$} is denoted by $N_F(X)$, that is, $N_F(X)=\{b\in B:\ \text{there exists an edge $ab\in F$ with $a\in X$}\}$. We will drop the subscript $F$ when $F$ is the whole edge set. We denote the \textbf{set of vertices in $X$ incident to edges in $F$} by $X(F)$. A forest (or tree) $F\subseteq E$ is a \textbf{$B_2$-forest} (or \textbf{$B_2$-tree}, respectively) if $d_F(b)=2$ for every $b\in B$. The existence of a $B_2$-forest was characterized by Lov\'asz~\cite{lovasz1970generalization}.

\begin{thm}[Lov\'asz]\label{thm:lovasz}
Let $G=(A,B;E)$ be a bipartite graph. Then there exists a $B_2$-forest in $G$ if and only if the strong Hall condition holds for every non-empty subset of $B$, that is,
\begin{equation*}
    |N(X)|\geq |X|+1\quad\text{for all $\emptyset\neq X\subseteq B$}.
\end{equation*}
\end{thm}

Matroids were introduced by Whitney \cite{whitney1992abstract} and independently by Nakasawa \cite{nishimura2009lost} as abstract generalizations of linear independence in vector spaces. A matroid $M$ is a pair $(S,\cI)$ where $S$ is the \textbf{ground set} of the matroid and $\cI\subseteq 2^S$ is the family of \textbf{independent sets} that satisfy the following, so-called \textbf{independence axioms}: (I1) $\emptyset\in\cI$, (I2) $X\subseteq Y\in \cI\Rightarrow X\in\cI$, and (I3) $X,Y\in\cI,\ |X|<|Y|\Rightarrow\exists e\in Y-X\ \text{s.t.}\ X+e\in\cI$.
The \textbf{rank} of a set $X\subseteq S$ is the maximum size of an independent subset of $X$ and is denoted by $r_M(X)$. The maximal independent sets of $M$ are called \textbf{bases}. A \textbf{cut} is an inclusionwise minimal set $X\subseteq S$ that intersects every basis. A \textbf{loop} is an element that is non-independent on its own. Two non-loop elements $e,f\in S$ are \textbf{parallel} if $\{e,f\}$ is non-independent. Given a matroid $M=(S,\cI)$, its \textbf{restriction} to a subset $S'\subseteq S$ is the matroid $M|_{S'}=(S',\cI')$ where $\cI'=\{I\in\cI:\ I\subseteq S'\}$. \textbf{Adding a parallel copy} of an element $s\in S$ results in a matroid $M'=(S',\cI')$ on ground set $S'=S+s'$ where $\cI'=\{X\subseteq S':\ \text{either}\ X\in\cI,\ \text{or}\ s\notin X,\ s'\in X\ \text{and}\ X-s'+s\in\cI\}$. The \textbf{dual} of $M$ is the matroid $M^*=(S,\cI^*)$ where $\cI^*=\{X\subseteq S:\ S-X\ \text{contains a basis of $M$}\}$. The \textbf{$k$-truncation} of a matroid $M=(S,\cI)$ is a matroid $(S,\cI_k)$ with $\cI_k=\{X\in\cI:|X|\leq k\}$. We denote the $k$-truncation of $M$ by $(M)_{k}$. The \textbf{direct sum} $M_1\oplus M_2$ of matroids $M_1=(S_1,\cI_1)$ and $M_2=(S_2,\cI_2)$ on disjoint ground sets is a matroid $M=(S_1\cup S_2,\mathcal{I})$ whose independent sets are the disjoint unions of an independent set of $M_1$ and an independent set of $M_2$. 
The \textbf{sum} $M_1+M_2$ of $M_1=(S,\cI_1)$ and $M_2=(S,\cI_2)$ on the same ground set is a matroid $M=(S,\cI)$ whose independent sets are the disjoint unions of an independent set of $M_1$ and an independent set of $M_2$. 

The rank function of the sum of $k$ matroids was characterized by Edmonds and Fulkerson \cite{edmonds1965transversals}. 

\begin{thm}[Edmonds and Fulkerson]
Let $M_1=(S,\cI_1),\dots,M_k=(S,\cI_k)$ be matroids on the same ground set $S$ with rank functions $r_1,\dots,r_k$, respectively. Then the maximum size of an independent set of $M_1+\dots+M_k$ is $\min\{\sum_{i=1}^k r_i(X)+|S-X|:\ X\subseteq S\}$.
\end{thm}

As a corollary, we get a characterization for the partitionability of the ground set into $k$ independent sets in a matroid.

\begin{cor}\label{cor:ef}
Let $M=(S,\cI)$ be a matroid with rank function $r$. Then $S$ can be partitioned into $k$ independent sets if and only if $|X|\leq k\cdot r(X)$ holds for every $X\subseteq S$. 
\end{cor}

Another corollary is a characterization for the existence of $k$ disjoint bases of a matroid. 

\begin{cor}\label{cor:ef2}
Let $M=(S,\cI)$ be a matroid with rank function $r$. Then $M$ has $k$ pairwise disjoint bases if and only if $|S-X|\geq k\cdot (r(S)-r(X))$ holds for every $X\subseteq S$. 
\end{cor}

\section{Transversal, graphic and paving matroids}\label{sec:easy}

As a warm-up, we first consider three basic cases: transversal, graphic, and paving matroids. Although the proofs are simple, they might help the reader to get familiar with the notion of reduction. Also, we show the connection to some earlier results such as Gallai colorings of complete graphs. 

\subsection{Transversal matroids} \label{sec:transversal}

Given a bipartite graph $G=(S,T;E)$, a set $X\subseteq S$ is \textbf{matchable} if there is a matching of $G$ covering $X$. The matchable sets satisfy the independence axioms; the matroid obtained this way is a called a \textbf{transversal matroid}. It is an easy exercise to show that the size of $T$ can be chosen to be $r$ where $r$ denotes the rank of the matroid (see e.g. \cite{frank2011connections}).  In Section~\ref{sec:gammoid}, we will use that the rank of a subset $X\subseteq S$ in the transversal matroid is $r(X)=\min\{|X|-|Y|+|N(Y)|:\ Y\subseteq X\}$ by the Frobenius-K\H{o}nig-Hall theorem \cite{konig1916graphen,hall1935representatives,frobenius1917zerlegbare}.

\thmtransversal*
\begin{proof}
Let $G=(S,T;E)$ a bipartite graph where $T=\{t_1,\dots,t_r\}$, $r$ being the rank of the transversal matroid on $S$. By assumption, the transversal matroid is $k$-colorable, so there exist $k$ matchings $F_1,\dots,F_k$ covering every vertex in $S$ exactly once. We may assume that none of these matchings is empty. Let $S_i=\bigcup_{j=1}^k N_{F_j}(t_i)$ for $i=1,\dots,r$ (see Figure~\ref{fig:transversal}). Then $S_1\cup\dots \cup S_r$ is a partition of $S$ with classes of size at most $k$. Pick an arbitrary element $s_j\in S_j$ for $j=1,\dots,r$. The edge set $\{t_js_j:\ j=1,\dots,r\}$ shows that the picked elements form a matchable set, hence the partition matroid defined by the partition is a $k$-colorable rank preserving reduction of the transversal matroid.
\end{proof}

\begin{figure}
    \centering
    \includegraphics[width=0.9\textwidth]{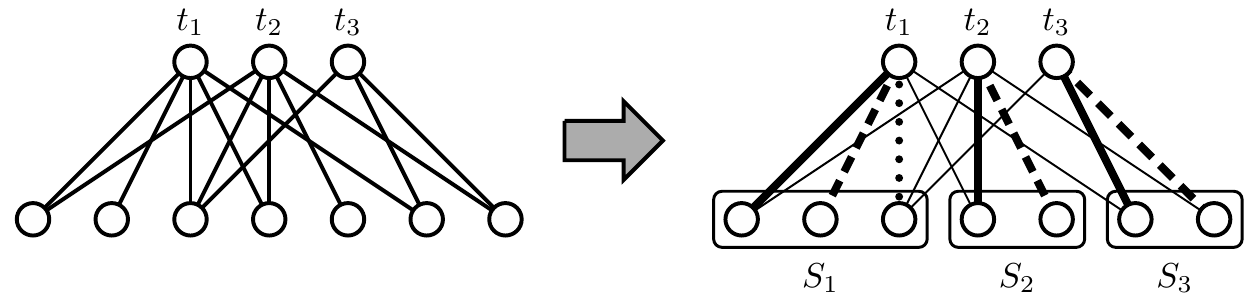}
    \caption{An illustration of the proof of Theorem~\ref{thm:transversal}. Thick, dashed and dotted edges are corresponding to three matchings covering $S$.}
    \label{fig:transversal}
\end{figure}

\subsection{Graphic matroids} \label{sec:graphic}

For a graph $G=(V,E)$, the \textbf{graphic matroid} $M=(E,\cI)$ of $G$ is defined on the edge set by considering a subset $F\subseteq E$ to be independent if it is a forest, that is, $\cI=\{F\subseteq E:\ F\ \text{does not contain a cycle}\}$. Nash-Williams \cite{nash1964decomposition} gave a characterization for $G$ being decomposable into $k$ forests, or in other words, for the graphic matroid of $G$ being $k$-colorable.

\begin{thm}[Nash-Williams] \label{thm:nw}
Given a graph $G=(V,E)$, the edge set can be decomposed into $k$ forests if and only if $|E[X]|\leq k(|X|-1)$ for every non-empty subset $X$ of $V$.
\end{thm}

\begin{figure}[h!]
\centering
\includegraphics[width=.65\linewidth]{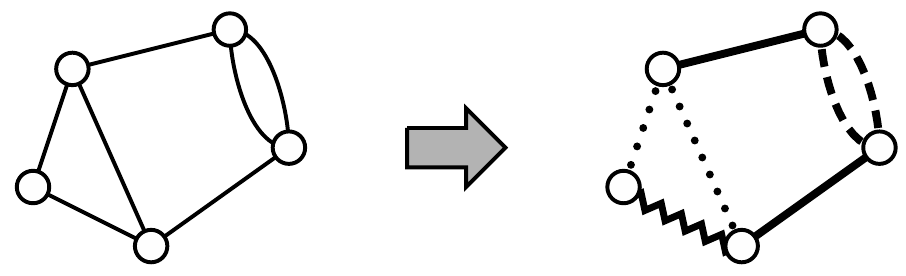}
\caption{An illustration of the proof of Theorem~\ref{thm:graphic}. The graph $G=(V,E)$ can be decomposed into three forests. Let $S_1,S_2,S_3$ and $S_4$ denote the sets of thick, dashed, dotted and zigzag edges, respectively. Then $S_{i+1}$ is a minimum cut in one of the components of $G-\bigcup_{j=1}^i S_j$ for $i=0,\dots,3$. Observe that there is no rainbow colored cycle in $G$ (in which any two edges receive different colors).}
\label{fig:graphic}
\end{figure}

\thmgraphic*
\begin{proof}
Let $G=(V,E)$ be a graph whose graphic matroid $M=(E,\cI)$ is $k$-colorable and let $K\subseteq V$ be a connected component of $G$ of size at least $2$. We claim that there exists a cut in $K$ of size at most $2k-1$. Indeed, if every cut of $K$ contains at least $2k$ edges then $K$ is a $2k$-edge-connected component and so $|E[K]|\geq k|K|$ by counting the edges around each vertex in $K$. By Theorem~\ref{thm:nw}, this contradicts the $k$-colorability of $M$.

Set $S_0:=\emptyset$ and $i:=0$. As long as there exists a connected component $K$ in $G-\bigcup _{j=0}^i S_j$ of size at least $2$, let $S_{i+1}\subseteq E$ be a minimum cut of $K$ (see Figure~\ref{fig:graphic}), and update $i:=i+1$. By the above, $|S_{i+1}|\leq 2k-1$. Let $E=S_1\cup\dots\cup S_q$ denote the partition thus obtained. We claim that the partition matroid corresponding to this partition is a reduction of $M$. In order to see this, we have to show that every cycle of $G$ intersects at least one of the partition classes in at least two elements. Given a cycle $C$, let $i$ be the smallest index with $|S_i\cap C|>0$. Then $C\subseteq\bigcup_{j\geq i}^q S_j$ and $S_i$ is a cut in $\bigcup_{j\geq i}^q S_j$, hence $|S_i\cap C|\geq 2$. As the deletion of $S_i$ increases the number of components of $G-\bigcup _{j=0}^{i-1} S_j$ by exactly one for $i=1,\dots,q$, the rank of the partition matroid thus obtained is the same as the rank of the graphic matroid of $G$, concluding the first half of the theorem.

To show that the given bound is tight, let $G=(V,E)$ be a complete graph on $2k$ vertices. By Nash-Williams' theorem, the coloring number of the graphic matroid of $G$ is $k$. Observe that reducing the graphic matroid of $G$ to a partition matroid is equivalent to coloring the edges of the graph in such a way that there is no cycle whose edges are colored with completely different colors. An edge coloring of a complete graph is called a \textbf{Gallai coloring} if no triangle is colored with three distinct colors,  which is a weaker restriction than the above. Bialostocki,  Dierker and Voxman \cite{bialostocki2001either} showed that every Gallai coloring contains a monochromatic spanning tree. This means that for any reduction of the graphic matroid of $G$ to a partition matroid, there is a partition class of size at least $2k-1$.
\end{proof}

\begin{rem}
Theorem~\ref{thm:graphic} can be proved in a similar way by observing that any graph that can be decomposed into $k$ forests contains a vertex of degree at most $2k-1$. The advantage of the proof based on cuts is twofold: it provides a rank preserving reduction, and it can be straightforwardly extended to arbitrary matroids in the following sense.

\begin{thm}\label{thm:gencut}
If $M=(S,\cI)$ is a matroid so that $M|_{S'}$ has a cut of size at most $k$ for any $S'\subseteq S$, then $M$ can be reduced to a $k$-colorable partition matroid.
\end{thm} 

The proof of Theorem~\ref{thm:gencut} is based on the fact that the intersection of a circuit and a cut in a matroid cannot have size $1$.
\end{rem}

\subsection{Paving matroids} \label{sec:paving}

A matroid $M=(S,\cI)$ of rank $r$ is called \textbf{paving} if every set of size at most $r-1$ is independent, or in other words, every circuit of the matroid has size at least $r$. 

 Although paving matroids have a very restricted structure and so are quite well-understood, they are playing a fundamental role among matroids. After Blackburn, Crapo and Higgs \cite{blackburn1973catalogue} enumerated all matroids up to eight elements, it was observed that most of these matroids are paving matroids. Crapo and Rota suggested that perhaps paving matroids dominate the enumeration of matroids \cite{crapo1970foundations}. This statement was made precise by Mayhew, Newman, Welsh and Whittle in \cite{mayhew2011asymptotic}. They conjectured that the asymptotic fraction of matroids on $n$ elements that are paving tends to $1$ as $n$ tends to infinity. Although this remains open, a similar statement on the asymptotic ratio of the logarithms of the numbers of matroids and sparse paving matroids has been proven in \cite{pendavingh2015number}.

First we consider paving matroids of arbitrary rank.

\thmpaving*
\begin{proof}
Consider any partition $S=S_1\cup\dots\cup S_{r-1}$ into $r-1$ parts of almost equal sizes, that is, $|S_i|=\lfloor|S|/(r-1)\rfloor$ or $|S_i|=\lceil|S|/(r-1)\rceil$ for $i=1,\dots,r-1$. As $M$ is $k$-colorable, we have $|S|\leq kr$ and so $|S_i|\leq \lceil kr/(r-1)\rceil$. As $M$ is paving, any set of size at most $r-1$ is independent, hence the partition matroid $N$ defined by the partition $S_1\cup\dots\cup S_{r-1}$ is a $\lceil kr/(r-1)\rceil$-colorable reduction of $M$, as required. 
\end{proof}

The bound on the coloring number can be improved when $r=2$, and the reduction can be chosen to be rank preserving.

\thmpavingtwo*
\begin{proof}
Let $S=T_1 \cup \dots \cup T_q$ denote the partition of the ground set into classes of parallel elements, that is, for every $x \in T_i$ and $y \in T_j$ the set $\{x,y\}$ is independent if and only if $i\neq j$. We may assume that $|T_1| \ge \dots \ge |T_q|$. Note that $|T_1|\leq k$ as the matroid is $k$-colorable. Let $i$ denote the smallest index such that $|T_1 \cup \dots \cup T_i| \ge |S|/3$ holds, and consider the partition $S = S_1 \cup S_2$ where $S_1 = T_1 \cup \dots \cup T_i$ and $S_2 = T_{i+1} \cup \dots \cup T_q$. If $i=1$, then $|S_1| = |T_1| \le k$, otherwise \[|S_1| = (|T_1|+\dots+|T_{i-1}|)+|T_i| < \frac{|S|}{3}+|T_i| \le \frac{|S|}{3} + |T_1| < \frac{2|S|}{3} \le \frac{4k}{3},\] where we used that $|S| \le 2k$ holds as $M$ is $k$-colorable and $r = 2$. By the definition of $i$, we have $|S_2| \le 2|S|/3 \le 4k/3$ as well. Thus $\max\{|S_1|, |S_2|\} \le 4k/3$ always holds, hence the partition matroid $N$ defined by the partition $S_1 \cup S_2$ is a $\lfloor 4k/3 \rfloor$-colorable reduction of $M$.

The bound $\lfloor 4k/3 \rfloor$ on the coloring number of $N$ is tight. Let $S$ be a set of size $2k$ and take a partition $S=S_1\cup S_2 \cup S_3$ where $\lceil |S|/3\rceil =|S_1| \ge |S_2| \ge |S_3| = \lfloor |S|/3 \rfloor$. Consider the laminar matroid $M=(S,\cI)$ defined by the laminar family $\{S,S_1,S_2,S_3\}$ where $X \subseteq S$ is independent if and only if $|X| \le 2$ and $|X \cap S_i| \le 1$ for $i=1,2,3$. It is not difficult to see that the coloring number of $M$ is $k$. Suppose that $M$ is reducible to a partition matroid $N$. The rank of $N$ is either 1 or 2, as $M$ has rank 2. In the former case $\rchi(N)=2k$, while in the latter case $N$ is defined by a partition $S=P_1\cup P_2$. Then every $S_i$ is a subset of either $P_1$ or $P_2$, as otherwise there exists two elements $x,y\in S_i$ such that $x \in P_1$ and $y\in P_2$, implying that $\{x,y\}$ is independent in $N$ but dependent in $M$, a contradiction. Thus $P_1$ or $P_2$ contains at least two of the $S_i$'s, and so has size at least $|S_2|+|S_3| =|S|-|S_1| = 2k-\lceil 2k/3 \rceil = \lfloor 4k/3\rfloor$, proving $\rchi(N)\ge \lfloor 4k/3 \rfloor$.
\end{proof}

For the proof of Theorem~\ref{thm:paving3}, we will need two technical statements. The first lemma describes the structure of paving matroids \cite{hartmanis1959lattice,welsh2010matroid,frank2011connections}.

\begin{lem}\label{lem:pav}
Let $r\geq 2$ be an integer and $S$ a set of size at least $r$. Let $\mathcal{H}=\{H_1,\dots,H_q\}$ be a (possibly empty) family of proper subsets of $S$ in which every set $H_i$ has at least $r$ elements and the intersection of any two of them has at most $r-2$ elements. Then the set system $\mathcal{B}_{\mathcal{H}}=\{X\subseteq S:\ |X|=r,\ X\not\subseteq H_i\ \text{for } i=1,\dots,q\}$ forms the set of bases of a paving matroid. Moreover, every paving matroid can be obtained in this form.
\end{lem}

The next lemma characterizes the coloring number of paving matroids.

\begin{lem}\label{lem:pavrank}
Let $\mathcal{H}=\{H_1,\dots, H_q\}$ be a (possibly empty) family satisfying the conditions of Lemma~\ref{lem:pav}, and let $M=(S,\cI)$ be the paving matroid defined by $\mathcal{H}$. Then 
\[\rchi(M) = \max\left\{\left\lceil\frac{|S|}{r}\right\rceil, \left\lceil\frac{|H_1|}{r-1}\right\rceil, \dots, \left\lceil\frac{|H_q|}{r-1}\right\rceil\right\}.\]
\end{lem}
\begin{proof}
Corollary~\ref{cor:ef} implies that $\rchi(M) = \max\{\lceil|X|/r(X) \rceil: \emptyset \ne X \subseteq S\}$. Since $r(S)=r$ and $r(H_i)=r-1$ (as every set of size at most $r-1$ is independent), we get that $\rchi(M)\ge \max \{\lceil |S|/r\rceil , \lceil |H_1|/(r-1)\rceil, \dots, \lceil |H_q|/(r-1) \rceil\}$.

To see the reverse inequality, take an arbitrary subset $X\subseteq S$. If $|X|\le r-1$, then $r(X)=|X|$ holds as the matroid is paving, therefore $|X|/r(X)=1$. If $|X|\ge r$ and $X\subseteq H_i$ for some $i$, then $r(X)=r-1$ and so $|X|/r(X) \le |H_i|/(r-1)$. Finally, if $|X|\ge r$ and none of the $H_i$'s contains $X$, then $r(X)=r$ and so $|X|/r(X) \le |S|/r$, proving our claim.
\end{proof}

Now we are ready to prove Theorem~\ref{thm:paving3}.

\thmpavingthree*
\begin{proof}
Let $\mathcal{H}=\{H_1,\dots, H_q\}$ be a (possibly empty) family satisfying the conditions of Lemma~\ref{lem:pav} that defines $M$. Without loss of generality, we may assume that $|H_1| \ge \dots \ge |H_q|$. We distinguish two cases.\\

\noindent \textbf{Case 1.} $|S|/r \le |H_1|/(r-1)$.

Consider any partition $H_1=S_1\cup\dots\cup S_{r-1}$ into $r-1$ parts of almost equal sizes, that is, $|S_i|=\lfloor|H_1|/(r-1)\rfloor$ or $|S_i|=\lceil|H_1|/(r-1)\rceil$ for $i=1,\dots,r-1$, and let $S_r = S-H_1$. Note that none of $S_1, \dots, S_r$ is empty since $H_1$ is a proper subset of $S$ of size at least $r-1$.
Taking any elements $s_1 \in S_1, \dots, s_r \in S_r$ the set $X=\{s_1, \dots, s_r\}$ is independent in $M$ as $X \not \subseteq H_i~(i=1,\dots, q)$ by $|X \cap H_1|=r-1$ and $|H_1\cap H_i| \le r-2 ~ (i=2,\dots, q)$.
Thus the partition matroid $N=(S,\cJ)$ defined by the partition $S_1\cup \dots \cup S_r$ is a rank preserving reduction of $M$.
$N$ is clearly $\rchi(M)$-colorable as $|S_i| \le \lceil |H_1|/(r-1)\rceil = \rchi(M)$ for $i=1,\dots, r-1$ and $|S_r|=|S|-|H_1| \le r|H_1|/(r-1)-|H_1|=|H_1|/(r-1) \le \rchi(M)$.\\

\noindent \textbf{Case 2.} $|S|/r > |H_1|/(r-1)$.

Pick an arbitrary $s\in S$, let $H_{i_1}, \dots, H_{i_l}$ denote the sets of the family $\mathcal{H}$ containing $s$ and let $H'_{j}=H_{i_j}-s$ for $j = 1,\dots, l$. The sets $H'_1, \dots, H'_l$ are disjoint as $|H_i\cap H_j| \le r-2=1$ for $i\ne j$. We may assume that $|H'_1|\ge \dots \ge |H'_l|$. Note that for any set $T\subseteq S-s$ which does not intersect any $H'_j$ properly, the partition $S=\{s\}\cup T \cup (S-T-s)$ defines a partition matroid $N=(S,\cJ)$ which is a reduction of $M$. 

If $|H'_1|+\dots +|H'_l| < |S|/3$, let $T\subseteq S-s$ be a set of size $\lfloor|S|/2\rfloor$ containing $H'_1\cup \dots \cup H'_l$. Then $\rchi(N)=\max\{|T|, |S|-|T|-1\} \le |S|/2 < 2|S|/3 \le 2\rchi(M)$.
If $|H'_1|+\dots+|H'_l|\ge|S|/3$, then let $j$ denote the smallest index such that $|H'_1|+\dots+|H'_j| \ge |S|/3$ and let $T=H'_1\cup \dots \cup H'_j$. If $j=1$, then $|T|=|H'_1|<2|S|/3$ by our assumption $|S|/r > |H_1|/(r-1)$ and $r = 3$. Otherwise $|H_1|<|S|/3$ and so $|T| \le |S|/3+|H'_j| \le |S|/3+|H'_1| < 2|S|/3$. Thus $\rchi(N) = \max\{|T|, |S|-|T|-1\} < 2|S|/3 \le 2\lceil |S|/3\rceil = 2\rchi(M)$.
\end{proof}

\begin{rem}
Note that Case 1 of the proof does not rely on the fact that $r=3$. That is, any paving matroid satisfying the assumption of Case 1 has a rank preserving reduction $N\preceq_r M$ with $\rchi(N)=\rchi(M)$.
\end{rem}

While Theorem~\ref{thm:paving3} provides a rank preserving reduction, Theorem~\ref{thm:paving} gives a better bound on the coloring number of the reduction for $r=3$. The bound $\lceil 3k/2\rceil$ is not necessarily tight. A computer-assisted case checking shows that the tight bound for $k=3$ is $4$, an extremal example being the Fano matroid. However, we show that $\lceil 3k/2\rceil$ is tight for infinitely many values of $k$.

A \textbf{finite projective plane} is a pair $(S,\mathcal{L})$, where $S$ is a finite set of \textbf{points} and $\mathcal{L}\subseteq 2^S$ is the family of \textbf{lines} that satisfies the following axioms: (P1) any two distinct points are on exactly one line, (P2) any two distinct lines have exactly one point in common, (P3) there exists four points, no three of which are collinear. For every projective plane there exists a number $q$ called the \textbf{order}, such that (1) each line in the plane contains $q+1$ points, (2) $q+1$ lines pass through each point of the plane, (3) the plane contains $q^2+q+1$ points and $q^2+q+1$ lines \cite{welsh2010matroid}.

The family of lines satisfies the conditions of Lemma~\ref{lem:pav}, thus every projective plane defines a paving matroid $M=(S,\mathcal{I})$ of rank 3. A partition matroid $N=(S,\mathcal{J})$ is a reduction of $M$ if and only if the coloring of $S$ defined by the partition classes of $N$ satisfies the conditions of the following theorem.

\begin{thm} \label{thm:proj}
Consider any 3-coloring of the points of a projective plane of order $q$ such that each line contains at most 2 colors. Then at least one of the following cases holds:
\begin{enumerate}[(i)]
	\item there exists an empty color class,
	\item there exists a color class of size 1,
	\item one of the color classes is the complement of a line.
\end{enumerate}
\end{thm}
\begin{proof}
Let $1$, $2$ and $3$ denote the three colors and $S_1$, $S_2$, $S_3$ the corresponding color classes. The proof is based on the following claim.

\begin{cl} \label{cl:line}
There exists a color class which is a subset of a line.
\end{cl}
\begin{proof}
Suppose indirectly that each of the three color classes contains three non-collinear points. Pick arbitrary points $p_1, p_2, p_3$ from color classes $S_1, S_2, S_3$, respectively. Let $L_i$ denote the line through $p_{i+1}$ and $p_{i+2}$, and set $m_{i,i+1}=|L_i \cap S_{i+1}|$ and $m_{i,i+2} = |L_i \cap S_{i+2}|$ (all indices are meant in a cyclic order). As every line of the plane has $q+1$ points, we have $m_{i,i+1}+m_{i,i+2} = q+1$.
Each line through a fixed point of $S_i$ has exactly one common point with the line $L_i$, hence $m_{i,i+1}$ of them contain colors $i$ and $i+1$ and $m_{i,i+2}$ of them contain colors $i$ and $i+2$. Since $p_1, p_2, p_3$ were arbitrary, we get that each line containing colors $i+1$ and $i+2$ has $m_{i,i+1}$ points of color $i+1$ and $m_{i,i+2}$ points of color $i+2$. 

As $S_{i+1}$ contains three non-collinear points, there exists a point $p'_{i+1} \in S_{i+1}-L_i$. By changing $i$ to $i+1$ in the previous paragraph, we get that exactly $m_{i+1, i+2}$ lines through $p’_{i+1}$ contain colors $i+1$ and $i+2$. As the $m_{i,i+2}$ lines through $p'_{i+1}$ and one of the points of $L_i \cap S_{i+2}$ contain colors $i+1$ and $i+2$, and the number of lines through $p'_{i+1}$ with these colors is $m_{i+1,i+2}$, we get that $m_{i,i+2} \le m_{i+1,i+2}$. By symmetry, we obtain $m_{i+1,i+2}=m_{i,i+2}$. Therefore \[m_{1,2}=m_{3,2} = q+1-m_{3,1} = q+1-m_{2,1} = m_{2,3}=m_{1,3}=q+1-m_{1,2},\] hence $m_{i,i+1}=m_{i,i+2}=(q+1)/2$ for all $i$. We get that all lines through $p_{i+1}$ contain $(q+1)/2$ points of color $i$, hence $|S_i|=(q+1)^2/2$. Therefore $|S_1|+|S_2|+|S_3| = 3(q+1)^2/2 > q^2+q+1$, a contradiction.
\end{proof}

By Claim~\ref{cl:line}, we may assume that $S_1 \subseteq L$ for a line $L$. Suppose indirectly that none of the cases \textit{(i)}, \textit{(ii)} and \textit{(iii)} hold. As $|S_1|\ge 2$, we can pick two distinct points $p_1, p'_1 \in S_1$. As none of $S_2$ and $S_3$ is the complement of $L$, there exists $p_2 \in S_2-L$ and $p_3 \in S_3-L$. The points of the line through $p_1$ and $p_2$ have color 1 or 2 and the points of the line through $p'_1$ and $p_3$ have color 1 or 3, hence the intersection of these lines have color 1. This intersection point cannot lie on $L$, hence $S_1 \not \subseteq L$, a contradiction. 
\end{proof}

\begin{cor}\label{cor:pav}
Let $M=(S,\cI)$ be a paving matroid of rank $3$ defined by the lines of a projective plane of order $q$. Suppose that $N=(S,\cJ)$ is a partition matroid such that $N \preceq M$. Then \[\rchi(N) \ge \begin{cases} (|S|-1)/2, & \text{if } q \text{ is odd,} \\ (|S|+1)/2, & \text{if } q \text{ is even.} \end{cases}\]
In particular, if $q\equiv 4  \pmod{6}$ then $\rchi(N) \ge \left \lceil \frac{3\rchi(M)}{2}\right \rceil$, and if equality holds then $N$ is not a rank preserving reduction of $M$. 
\end{cor}
\begin{proof}
As $N\preceq M$, the coloring defined by the partition classes of $N$ satisfies the conditions of Theorem~\ref{thm:proj}. If there exists an empty color class, then one of the color classes has size at least $\lceil |S|/2 \rceil = (|S|+1)/2$. If there exists a color class containing only one point $p$, then each of the $q+1$ lines through $p$ are monochromatic except for $p$, hence one of the color classes has size at least $q\lceil (q+1)/2\rceil$, that is, $q(q+1)/2 = (|S|-1)/2$ if $q$ is odd, and $q(q+2)/2=(|S|+q-1)/2$ if $q$ is even. If one of the color classes is the complement of a line, then it has size $q^2>(|S|+1)/2$. In all three cases, we get a color class of size at least $(|S|-1)/2$ if $q$ is odd, and $(|S|+1)/2$ if $q$ is even, proving our bound on the coloring number of $N$.

Assume now that $q\equiv 4 \pmod6$. Lemma~\ref{lem:pavrank} implies that $\rchi(M) = \max\{\lceil(q^2+q+1)/3\rceil, \lceil(q+1)/2\rceil\} = (q^2+q+1)/3$, and so $\lceil 3\rchi(M)/2\rceil = \lceil(q^2+q+1)/2\rceil = (|S|+1)/2 \le \rchi(N)$. If equality holds then $N$ has rank 2, that is, one of the color classes is empty, since we have strict inequalities above in the other two cases.
\end{proof}

Corollary~\ref{cor:pav} implies that the bound $\lceil 3k/2 \rceil$ for paving matroids of rank 3 is tight for infinitely many values of $k$. Indeed, consider projective planes of order $q=4^\ell$ for $\ell\in\mathbb{Z}_{>0}$ and set $k=\frac{q^2+q+1}{3}$.

\section{Gammoids} \label{sec:gammoid}

The aim of this section is to prove the main result of the paper, Theorem~\ref{thm:gammoid}. A generalization of transversal matroids can be obtained with the help of directed graphs. Given a directed graph $D=(V,A)$ and two sets $X,Y\subseteq V$, we say that $X$ is \textbf{linked} to $Y$ if $|X|=|Y|$ and there exists $|X|$ vertex-disjoint directed paths from $X$ to $Y$. Let $S\subseteq V$ be a set of starting vertices and $T\subseteq V$ be a set of destination vertices. Then the family $\cI=\{Y\subseteq T:\ \exists X\subseteq S\ \text{s.t. $X$ is linked to $Y$}\}$ forms the independent sets of a matroid that is called a \textbf{gammoid}. The gammoid is a \textbf{strict gammoid} if $T=V$. That is, a gammoid is obtained by restricting a strict gammoid to a subset of its elements.

Transversal matroids and gammoids are closely related. Ingleton and Piff \cite{ingleton1973gammoids} showed that strict gammoids are exactly the duals of transversal matroids, hence every gammoid is the restriction of the dual of a transversal matroid. 

\thmgammoid*
\begin{proof}
Let $M=(S,\cI)$ be a $k$-colorable gammoid where $k\geq 2$. By the result of Ingleton and Piff, $M$ can be obtained as the restriction of the dual of a transversal matroid. Let $\wtm$ be such a transversal matroid, and choose $\wtm$ in such a way that its rank is as small as possible. Let $G=(A,B;E)$ be a bipartite graph defining $\wtm$ with $S\subseteq A$ and $|B|$ being the rank of $\wtm$.

The high-level idea of the proof is the following. First we show that there exists a $B_2$-forest $F$ in $G$. Then, by using an alternating structure on the components of $F$, we prove that $F$ can be chosen in such a way that every component contains at most $2k-2$ vertices from $S$. Let $\cC$ denote the set of the connected components of $F$, and let $N=(S,\cJ)$ be the partition matroid corresponding to partition classes $S(C)$ for $C\in\cC$. Every component $C$ is a $B_2$-tree, hence it contains a perfect matching between $B(C)$ and $A(C)-a$ for any $a\in A(C)$. That is, if we leave out exactly one vertex from $A(C)$ for each $C\in\cC$, the remaining vertices of $A$ form a basis of $\wtm$, and so the set of deleted vertices form a basis in the strict gammoid that is the dual of $\wtm$. This implies that $N\preceq M$ with $\rchi(N)\leq 2k-2$, thus proving the theorem.

We start with an easy observation.

\begin{cl}\label{cl:kmatch}
$G$ contains $k$ matchings of size $|B|$ such that every vertex in $S$ is covered by at most $k-1$ of them. 
\end{cl}
\begin{proof}
Observe that a set $X\subseteq S$ is independent in $M$ if and only if $A-X$ contains a basis of $\wtm$, that is, $G-X$ has a matching covering $B$. The assumption that $M$ is $k$-colorable is equivalent to the condition that $S$ can be partitioned into $k$ independent sets of $M$, and the claim follows.
\end{proof} 

The following claim proves an inequality that we will rely on.

\begin{cl}\label{cl:cond}
$k\cdot(|A|-|B|)-|S-X|\geq k\cdot\max\{|Y|-|N(Y)|:\ Y\subseteq X\}$ for every $X\subseteq A$.
\end{cl}
\begin{proof}
Let $\wtmm$ be the matroid that is obtained from $\wtm$ by adding $k-1$ parallel copies of every element in $A-S$, and adding $k-2$ parallel copies of every element in $S$. The ground set $A'$ of $\wtmm$ has size $(k-1)|S|+k|A-S|$. Then Claim~\ref{cl:kmatch} states that $\wtmm$ has $k$ pairwise disjoint bases. 

Let $X\subseteq A$ be an arbitrary set and let $X'$ be the set consisting of all the parallel copies of the elements of $X$. Then $|X'|=(k-1)\cdot|X\cap S|+k\cdot|X-S|$ and $r_{\wtmm}(X')=r_{\wtm}(X)=\min\{|X|-|Y|+|N(Y)|:\ Y\subseteq X\}$. Recall that $|A'|=(k-1)\cdot|S|+k\cdot |A-S|$ and $r_{\wtmm}(A')=|B|$, hence
\begin{align*}
|A'|-|X'|
{}&{}=
(k-1)\cdot|S|+k\cdot |A-S|-(k-1)\cdot|X\cap S|-k\cdot|X-S|\\
{}&{}=
(k-1)\cdot |A|+|A-S|-(k-1)\cdot|X|-|X-S|\\
{}&{}=
(k-1)\cdot |A-X|+|A-S-X|, 
\end{align*}
and
\begin{align*}
r_{\wtmm}(A')-r_{\wtmm}(X')
{}&{}=
|B|-\min\{|X|-|Y|+|N(Y)|:\ Y\subseteq X\}\\
{}&{}=
|B|-|X|+\max\{|Y|-|N(Y)|:\ Y\subseteq X\}.
\end{align*}
By Corollary~\ref{cor:ef2} and Claim~\ref{cl:kmatch}, $|A'|-|X'|\geq k\cdot (r_{\wtmm}(A')-r_{\wtmm}(X'))$, thus we get 
\begin{align*}
(k-1)\cdot |A-X|+|A-S-X|
{}&{} \geq 
k\cdot(|B|-|X|+\max\{|Y|-|N(Y)|:\ Y\subseteq X\}).
\end{align*}
After rearranging, we obtain
\begin{align*}
k\cdot(|A|-|B|)-|S-X|
{}&{} \geq 
k\cdot \max\{|Y|-|N(Y)|:\ Y\subseteq X\}
\end{align*}
as stated.
\end{proof} 

Our next goal is to show that there exists a $B_2$-forest in $G$.

\begin{cl}\label{cl:b2}
$G=(A,B;E)$ contains a $B_2$-forest.
\end{cl}
\begin{proof}
As $G$ has a matching of size $|B|$, the Hall condition holds for every subset of $B$, thus $|N(U)|\geq|U|$ for every $U\subseteq B$. Let us call a set $U\subseteq B$ \textbf{tight} if $|N(U)|=|U|$. Assume that $G$ does not have a $B_2$-forest. Then, by Theorem~\ref{thm:lovasz}, there exists a non-empty tight set in $B$. For arbitrary tight sets $U,W\subseteq B$, we get
\begin{align*}
    |U|+|W|
    {}&{}=
    |N(U)|+|N(W)|
    = 
    |N(U)\cap N(W)|+|N(U)\cup N(W)|\\
    {}&{}\geq
    |N(U\cap W)|+|N(U\cup W)|
    \geq |U\cap W|+|U\cup W|\\
    {}&{}= |U|+|W|,
\end{align*}
hence equality holds throughout, and so $U\cap W$ and $U\cup W$ are also tight. This implies that there is a unique maximal tight set $\emptyset\neq Z\subseteq B$.

Let $X=A-N(Z)$. As $Z$ is a tight set, $\max\{|Y|-|N(Y)|:\ Y\subseteq X\}\geq |X|-|N(X)|\geq |A-N(Z)|-|B-Z|=|A|-|B|$, thus $S-X=N(Z)\cap S=\emptyset$ by Claim~\ref{cl:cond}. Furthermore, every matching of size $|B|$ provides a perfect matching between $Z$ and $N(Z)$. That is, $\wtm$ is the direct sum of the transversal matroids $\wtm'$ and $\wtm''$ defined by $G[Z\cup N(Z)]$ and $G[(B-Z)\cup(A-N(Z))]$, respectively. Therefore $M$ is the restriction of the dual of $\wtm''$ to $S$, contradicting the minimal choice of $\wtm$.  
\end{proof}

Take an arbitrary $B_2$-forest $F$ in $G$. We will need the following technical claim.

\begin{cl}\label{cl:leaf}
Every leaf of $F$ is in $S$.
\end{cl}
\begin{proof}
Suppose to the contrary that $F$ has a leaf vertex $a\in A-S$. Let $b\in B$ be the unique neighbor of $a$ in $F$, and let $G'=G-\{a,b\}$ denote the graph obtained by deleting vertices $a$ and $b$ form $G$. Let $M'=(S,\cI')$ denote the restriction of the dual of the transversal matroid defined by $G'$ to $S$. As the strong Hall condition holds for $G$, the maximum size of a matching of $G'$ is $|B|-1$. We claim that $M=M'$, contradicting the minimality of $G$.

Take an arbitrary set $X\in\cI'$. By definition, $G'-X$ has a matching $P'$ covering $B-b$. Then $P'+ab$ is a matching of $G-X$ covering $B$, showing that $\cI'\subseteq \cI$.

To see the opposite direction, consider any set $X\in\cI$. By definition, $G-X$ has a matching $P$ covering $B$. Take an arbitrary matching $P'$ of $G'$ covering $B-b$. Now $|P|=|B|=|B-b|+1=|P'|+1$, hence the symmetric difference $P\triangle P'$ contains an alternating path $Q$ whose first and last edges are in $P$, and one of the end vertices of $Q$ is $b$. Then $P\triangle Q$ is a matching of $G'-X$ covering $B-b$, implying $X\subseteq\cI'$.   
\end{proof}

We denote the difference $|A|-|B|$ by $q$. As $M$ is the restriction of $\wtm$ to $S$, $r_M(S)\leq q$ is clearly satisfied. Moreover, equality holds since, by Claim~\ref{cl:leaf}, every leaf of $F$ is in $S$, and taking an arbitrary leaf in every component of $F$ results a basis of $M$.

Let $\cC$ denote the set of connected components of $F$. Note that the forest might have components consisting of a single vertex of $A$. We have $|\cC|=|A|-|B|=q$ as $|A(C)|=|B(C)|+1$ for each $C\in\cC$. We call a component $C\in\cC$ \textbf{large} if $|S(C)|\geq 2k-1$, \textbf{normal} if $k\leq |S(C)|\leq 2k-2$, and \textbf{small} if $|S(C)|\leq k-1$. We say that a component $C' \in \cC$ is \textbf{reachable} from a component $C'' \in \cC$ if there exists an alternating sequence $C_1,b_1a_2,C_2,b_2a_3,\dots,b_{p-2}a_{p-1},C_{p-1},b_{p-1}a_p,C_p$ of components and edges such that $C_1=C''$, $C_p=C'$, and $b_i\in B(C_i)$, $a_{i+1}\in A(C_{i+1})$ hold for $i=1,\dots,p-1$. Such an alternating sequence is called a \textbf{path}, the \textbf{length} of the path being $p-1$. The \textbf{distance} of $C'$ from $C''$ is the minimum length of a path from $C''$ to $C'$.

 We define a potential function on the set of $B_2$-forests as follows. Let $\nu\gg\mu_1\gg\lambda_1\gg\mu_2\gg\lambda_2\gg\dots\gg\mu_{q-1}\gg\lambda_{q-1}$ be a decreasing sequence of $2q-1$ positive numbers such that the ratio between any two consecutive ones is at least $|A|+2$. Recall that $|\cC|=q$. For a component $C\in\cC$, the minimum distance of $C$ from a large component is denoted by $\dist(C)$. We define $\dist(C)$ to be $+\infty$ if $C$ is not reachable from any of the large components. The \textbf{potential} of the $B_2$-forest $F$ is defined as
\begin{align*}
\varphi(F)
{}&{}=
\nu\cdot\sum_{C\in\cC} \max\{|S(C)|-(2k-2),0\}\hspace{-10mm}&\text{(total violation)}\\
{}&{}~~-
\sum_{i=1}^{q-1}\mu_i\cdot|\{C\in\cC:\ \dist(C)=i\}|&\text{(number of components at distance $i$)}\\
{}&{}~~+
\sum_{i=1}^{q-1}\lambda_i\cdot\sum_{\substack{C\in\cC \\ \dist(C)=i}}|S(C)|. &\text{(number of $S$-vertices in components at distance $i$)}
\end{align*}

Let $F$ be a $B_2$-forest for which $\varphi(F)$ is as small as possible. The following claim concludes the proof of the theorem.

\begin{cl}\label{cl:potential}
$F$ has no large components.
\end{cl}
\begin{proof}
Suppose indirectly that there exists a large component. By applying Claim~\ref{cl:cond}
with $X=\emptyset$, $|S|\leq k\cdot(|A|-|B|)=k\cdot |\cC|$, hence, by the pigeonhole principle, there exists a small component as well.

First we show that there exists a small component that is reachable from a large component. Suppose indirectly that this is not true, and let $\cC'\subseteq\cC$ denote the set of components that are not reachable from a large component. Note that $\cC'$ consists of normal and small components. Define $X=\bigcup\{ A(C):\ C\in\cC'\}$. By the definition of reachability, $N(X)=\bigcup\{B(C):\ C\in\cC'\}$ and so $|X|-|N(X)|=|\cC'|$. As every component in $\cC-\cC'$ is either normal or large and there is at least one large component, $|S-X|\geq k\cdot|\cC-\cC'|+1$. Then
\begin{align*}
k\cdot\max\{|Y|-|N(Y)|:\ Y\subseteq X\}
{}&{}\geq
k\cdot(|X|-|N(X)|)\\
{}&{}=
k\cdot|\cC'|\\
{}&{}=
k\cdot(|\cC|-|\cC-\cC'|)\\
{}&{}=
k\cdot(|A|-|B|)-k\cdot|\cC-\cC'|\\
{}&{}\geq
k\cdot(|A|-|B|)-|S-X|+1,
\end{align*}
contradicting Claim~\ref{cl:cond}.

\begin{figure}[t!]
\centering
\begin{subfigure}[t]{0.47\textwidth}
  \centering
  \includegraphics[width=.95\linewidth]{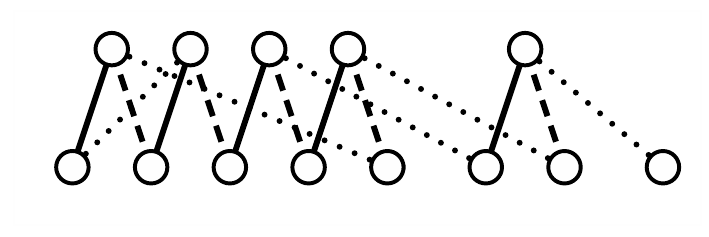}
  \caption{A graph $G=(A,B;E)$ with three matchings of size $|B|$ such that every vertex in $A$ is covered by at most two of them.}
  \label{fig:gammoid1}
\end{subfigure}\hfill
\begin{subfigure}[t]{0.45\textwidth}
  \centering
  \includegraphics[width=.95\linewidth]{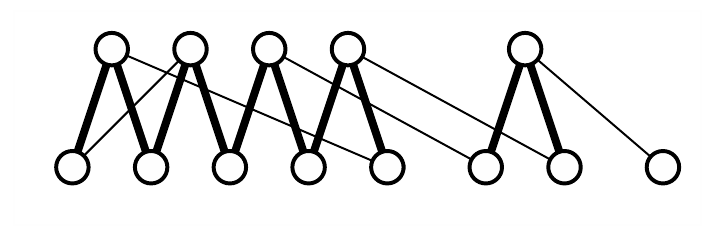}
  \caption{A $B_2$-forest $F$ of $G$. For simplicity, every component of $F$ is chosen to be a path.}
  \label{fig:gammoid2}
\end{subfigure}\vspace{5pt}
\begin{subfigure}[b]{0.45\textwidth}
  \centering
  \includegraphics[width=.95\linewidth]{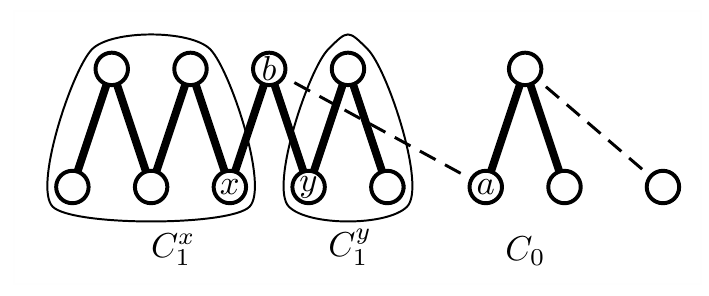}
  \caption{Alternating structure on the connected components of $F$.}
  \label{fig:gammoid3}
\end{subfigure}\hfill
\begin{subfigure}[b]{0.47\textwidth}
  \centering
  \includegraphics[width=.95\linewidth]{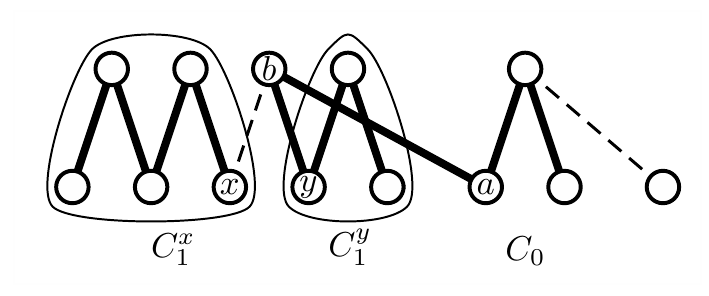}
  \caption{The $B_2$-forest we obtain by substituting $C_0$ and $C_1$ by $C_0+ab+by+C^y_1$ and $C^x_1$.}
  \label{fig:gammoid4}
\end{subfigure}
\caption{An illustration of the proof of Theorem~\ref{thm:gammoid}. In the example, $k=3$ and $S=A$. The only large component of $F$ is $C_1$, all the other components are small.}
\label{fig:gammoid}
\end{figure}

Let $C_0$ be a small component with $\dist(C_0)$ being minimal. By the above, $\dist(C_0)<+\infty$. Consider a shortest path from the set of large components to $C_0$, and let $C_1$ be the last component on the path before $C_0$. By the definition of a path, there exists an edge $ab$ with $a\in A(C_0)$ and $b\in B(C_1)$. Let $x,y\in A(C_1)$ denote the neighbors of $b$ in $C_1$. The deletion of $b$ from $C_1$ results in two connected components $C^x_1$ and $C^y_1$ such that $x \in C^x_1$ and $y \in C^y_1$ (see Figure~\ref{fig:gammoid}).

Assume first that $C_1$ is a large component. As $|S(C_1)|\geq 2k-1$ and $|S(C_0)|\leq k-1$, either $|S(C_0+ab+bx+C^x_1)|<|S(C_1)|$ or $|S(C_0+ab+by+C^y_1)|<|S(C_1)|$ by Claim~\ref{cl:leaf}. Hence substituting $C_0$ and $C_1$ either by $C_0+ab+bx+C^x_1$ and $C^y_1$ or by $C_0+ab+by+C^y_1$ and $C^x_1$ decreases the total violation in $\varphi(F)$, a contradiction.

Therefore $C_1$ is a normal component, and there is another non-small component $C_2$ before $C_1$ on the shortest path from the set of large components to $C_0$, together with an edge $b'a'$ with $a'\in A(C_1)$ and $b'\in A(C_2)$. We may assume that $a'\in C^x_1$.  We distinguish two cases.\\

\noindent\textbf{Case 1.} $|S(C^x_1)|\geq|S(C_1)|-|S(C_0)|$

Modify $F$ by substituting components $C_0$ and $C_1$ by $C_0+ab+by+C^y_1$ and $C^x_1$, respectively. By the assumption, $|S(C_0+ab+by+C^y_1)|=|S(C_0)|+|S(C_1)|-|S(C^x_1)|\leq 2\cdot|S(C_0)|\leq 2k-2$, thus no new large component appears. Furthermore, the set of components with distance less than $\dist(C_1)$ does not change. The distance of $C^x_1$ remains $\dist(C_1)$ because of the edge $b'a'$. 
If the distance of $C_0+ab+by+C^y_1$ is $\dist(C_1)$, then the number of components at distance $\dist(C_1)$ increases.
Otherwise, the distance of $C_0+ab+by+C^y_1$ is at least $\dist(C_1)+1$, hence the number of $S$-vertices in components at distance $\dist(C_1)$ decreases by Claim~\ref{cl:leaf}. In both cases, $\varphi(F)$ decreases, a contradiction.\\

\noindent\textbf{Case 2.} $|S(C^x_1)|<|S(C_1)|-|S(C_0)|$

Modify $F$ by substituting components $C_0$ and $C_1$ by $C_0+ab+bx+C^x_1$ and $C^y_1$, respectively. By the assumption, $|S(C_0+ab+bx+C^x_1)|\leq|S(C_0)|+|S(C_1)|-|S(C_0)|=|S(C_1)|$. As $C_1$ is normal, no new large component appears.  Furthermore, the set of components with distance less than $\dist(C_1)$ does not change. The distance of $C_0+ab+bx+C^x_1$ remains $\dist(C_1)$ because of the edge $b'a'$. The distance of $C^y_1$ is either $\dist(C_1)$ or $\dist(C_0)$. In the former case, the number of components at distance $\dist(C_1)$ increases, while in the latter case, the number of $S$-vertices in components at distance $\dist(C_1)$ decreases as $|S(C_1^x)| + |S(C_0)| < |S(C_1)|$. In both cases, $\varphi(F)$ decreases, a contradiction.
\end{proof}

By Claim~\ref{cl:potential}, $F$ has no large component. As we have seen before, the partition matroid $N=(S,\cJ)$ corresponding to partition classes $S(C)$ for $C\in\cC$ is a reduction of the original gammoid $M$ with coloring number at most $2k-2$. By Claim~\ref{cl:leaf}, $S(C)$ is non-empty for every $C\in\cC$ and $r_M(S)=|\cC|=q$, hence the reduction is rank preserving.

The bound on the coloring number of $N$ is tight. Consider the laminar matroid $M=(S,\cI)$ defined by the laminar family $\{S,S_1,\dots,S_k\}$ where $S_1\cup\dots\cup S_k$ is a partition of $S$ into subsets of size $k-1$. That is, the size of the ground set $S$ is $k^2-k$. We define a set $X\subseteq S$ to be independent in $M$ if $|X\cap S_i|\leq 1$ for $i=1,\dots,k$, and $|X|\leq k-1$. It is not difficult to see that $M$ is a strict gammoid with coloring number $k$.

We claim that if $N\preceq M$ is a partition matroid, then $\rchi(N)\geq 2k-2$. Let $P_1\cup \dots\cup P_q$ denote the partition defining $N$. Then every $S_i$ is a subset of some $P_j$, as otherwise there exists two elements $x,y\in S_i$ such that $x\in P_a$ and $y\in P_b$ for $a\neq b$, implying that $\{x,y\}$ is independent in $N$ but dependent in $M$, a contradiction. As the rank of $M$ is $k-1$, we have $q\leq k-1$. By the above, there exists a class $P_j$ that contains at least two of the $S_i$'s, and so has size at least $2k-2$, proving $\rchi(N)\geq 2k-2$.
\end{proof}

For the first sight, the proof seems to provide a polynomial-time algorithm for determining the partition matroid, assuming that a digraph $D=(V,A)$ representing the gammoid is given. A bipartite graph $G=(A,B;E)$ representing $\wtm$ can be constructed from $D$ (see e.g. \cite{frank2011connections}). The reductions appearing in the proofs of Claims~\ref{cl:b2} and \ref{cl:leaf} can be performed in polynomial time, hence we may assume that $G$ contains a $B_2$-forest $F$. Such a forest can be found by \cite{lovasz1970generalization}. By using the alternating structure described in the proof of Claim~\ref{cl:potential}, we can modify $F$ to get a $B_2$-forest in which every component contains at most $2k-2$ vertices from $S$. However, it is not clear how to bound the number of augmentation steps as the coefficients in the potential function can be exponential. An interesting question is whether this procedure terminates after a polynomial number of steps.

\section{Truncation and reducibility}
\label{sec:trun}

\thmtruncation*
\begin{proof}
Let $M=(S,\cI)$ denote a matroid of rank $r$ that is reducible to a $2\rchi(M)$-colorable partition matroid $N$.
As every $k$-truncation of $M$ can be obtained by a series of $r-1, r-2, \dots, k$-truncations, it suffices to prove that the $(r-1)$-truncation $M'$ of $M$ is reducible to a $2\rchi(M')$-colorable partition matroid.

Let $S=S_1 \cup \dots \cup S_q$ denote the partition that defines $N$. We may assume that $|S_1| \ge \dots \ge |S_q|$.
If $q \le r-1$, then $N$ is already a $2\rchi(M)$-colorable reduction of $M'$ and the claim follows by $\rchi(M') \ge \rchi(M)$. Hence assume that $q=r$.
Consider the partition matroid $N'$ defined by the partition classes $S_1, S_2, \dots, S_{r-2}, S_{r-1}\cup S_r$.
Then $N'$ is a reduction of $M'$, hence it is sufficient to prove that $N'$ is $2\rchi(M')$-colorable.

If $|S_{r-1}|+|S_r| \le |S_1|$, then $\rchi(N')=|S_1|=\rchi(N) \le 2\rchi(M) \le 2\rchi(M')$. Otherwise $|S_{r-1}|+|S_r| > |S_1|$, and so $\rchi(N')=|S_{r-1}|+|S_r|$.
Using $|S|=|S_1|+\dots + |S_r|$ and $|S_i| \ge (|S_{r-1}|+|S_r|)/2$ for $i=1, 2, \dots, r-2$, we get 
\begin{align*}
|S|
{}&{}\ge
(r-1)\cdot \frac{|S_{r-1}|+|S_r|}{2}+|S_{r-1}|+|S_r|\\
{}&{}=
\frac{r+1}{2}\cdot (|S_{r-1}|+|S_r|)\\
{}&{}=\frac{r+1}{2}\cdot \rchi(N').
\end{align*}
That is, \[\rchi(N')\le \frac{2|S|}{r+1} < 2\cdot \frac{|S|}{r-1} \le 2\rchi(M'),\] concluding the proof of the theorem.
\end{proof}

\begin{rem}
Note that an analogous statement holds if we replace $2\rchi(M)$ by $2\rchi(M)-1$, as we proved $\rchi(N') < 2\rchi(M')$ in the second case. As laminar matroids can be obtained from free matroids by taking direct sums and truncations, Theorem~\ref{thm:truncation} provides a simple proof that every $k$-colorable laminar matroid is reducible to a $(2k-1)$-colorable partition matroid. As laminar matroids form a subclass of gammoids, Theorem~\ref{thm:gammoid} implies that the bound can be improved to $2k-2$. However, it is not clear whether the analogue of Theorem~\ref{thm:truncation} holds if we replace $2\rchi(M)$ by $2\rchi(M)-2$.
\end{rem}

\section{Reduction to strongly base orderable matroids} \label{sec:sbo}

Edmonds and Fulkerson \cite{edmonds1965transversals} characterized the existence of $k$ disjoint bases in a single matroid. For the case of two matroids $M_1=(S,\cI_1)$ and $M_2=(S,\cI_2)$, Edmonds' celebrated matroid intersection theorem characterizes the maximum size of a common independent set \cite{edmonds1970submodular}. A fundamental problem of matroid optimization is to characterize the existence of a partitioning into $k$ common independent sets of two matroids. The importance of the problem is underpinned by a long list of well-studied conjectures that can be formalized as a special cases, such as Rota's beautiful conjecture on the rearrangements of bases \cite{huang1994relations}, or Woodall's conjecture on packing dijoins in a directed graph \cite{woodall1978menger}. Recently, the first two authors verified that the problem is difficult under the rank oracle model \cite{BS2019}. Although this result settles the complexity of the problem in general, it has no implications on its special cases. Hence finding matroid classes for which the problem becomes tractable is of interest.

There are only a few cases in which a proper characterization is known. These problems include the classical results of K\H{o}nig on $1$-factorization of bipartite graph \cite{konig1916graphen}, Edmonds' theorem on the existence of $k$ disjoint spanning arborescences of a digraph \cite{edmonds1973edge}, and the result of Keijsper and Schrijver on packing connectors \cite{keijsper1998packing}.

In general, there is a natural necessary condition for the existence of a partition into $k$ common independent sets: the ground set has to be partitionable into $k$ independent sets in both matroids. This condition is not sufficient in general. However, Davies and McDiarmid \cite{davies1976disjoint} observed that it is sufficient when both matroids are strongly base orderable. A matroid is \textbf{strongly base orderable} if for every two bases $B_1$ and $B_2$, there is a bijection $\gamma\colon B_1\rightarrow B_2$ with the property that $(B_1-X)\cup\gamma(X)$ is a basis for any $X\subseteq B_1$.

\begin{thm}[Davies and McDiarmid]\label{thm:dd}
Let $M_1=(S,\cI_1)$ and $M_2=(S,\cI_2)$ be strongly base orderable matroids. If $S$ can be partitioned into $k$ independent sets in both $M_1$ and $M_2$, then $S$ can be partitioned into $k$ common independent sets.
\end{thm}

Given a matroid $M=(S,\cI)$, an element $v \in S$ is said to be \textbf{$k$-spanned} if there are $k$ disjoint sets that span $v$. Kotlar and Ziv \cite{kotlar2005partitioning} proved that if $M_1$ and $M_2$ are matroids on $S$ and no element is 3-spanned in $M_1$ or $M_2$, then $S$ can be partitioned into $2$ common independent sets. They conjectured that this can be generalized to arbitrary $k$: if no element is $(k+1)$-spanned in $M_1$ or $M_2$, then $S$ can be partitioned into $k$ common independent sets. It is worth mentioning that if no element is $(k+1)$-spanned in a matroid then the matroid is $k$-colorable, hence the natural necessary condition is satisfied in this case. In \cite{takazawa2019generalized}, Takazawa and Yokoi proposed a new approach building upon the generalized-polymatroid intersection theorem. Their result gives a new interpretation of that of Kotlar and Ziv, and extends the list of those pairs of matroid classes for which a characterization is known for the existence of a partition into $k$ common independent sets. 

Aharoni and Berger \cite{aharoni2006intersection} proposed the following conjecture that would give the best possible upper bound for the minimum number of common independent sets of two matroids needed to cover the ground set.

\begin{conj}[Aharoni and Berger]\label{conj:ab2}
Let $M_1=(S,\cI_1)$ and $M_2=(S,\cI_2)$ be loopless matroids. Then $S$ can be partitioned into $\max\{\rchi(M_1),\rchi(M_2)\}+1$ common independent sets. 
\end{conj}

The idea of reducing a matroid to a partition matroid can be generalized to strongly base orderable matroids. Let $M_1$ and $M_2$ be arbitrary matroids on the same ground set, and assume that they are reducible to $k$-colorable strongly base orderable matroids $N_1$ and $N_2$, respectively. Then, by Theorem~\ref{thm:dd}, $S$ can be decomposed into $k$ common independent sets of $N_1$ and $N_2$. As $N_1\preceq M_1$ and $N_2\preceq M_2$, this gives a partition of $S$ into $k$ common independent sets of $M_1$ and $M_2$. 

In particular, the following statement strengthens Conjecture~\ref{conj:ab2}.

\begin{conj}\label{conj:sbo}
Let $M=(S,\cI)$ be a $k$-colorable matroid. Then $M$ has a $(k+1)$-colorable strongly base orderable reduction.  
\end{conj} 
Indeed, if Conjecture~\ref{conj:sbo} is true, then $M_1$ and $M_2$ have $(\rchi(M_1)+1)$ and $(\rchi(M_2)+1)$-colorable strongly base orderable reductions $N_1$ and $N_2$, respectively. By Theorem~\ref{thm:dd}, $S$ can be decomposed into $\max\{\rchi(M_1),\rchi(M_2)\}+1$ common independent sets, thus proving Conjecture~\ref{conj:ab2}. Although we do not expect Conjecture~\ref{conj:sbo} to hold in general, it might help to identify special cases for which the Aharoni--Berger conjecture holds.

\begin{rem}
The bound of $k+1$ for the strongly base orderable reduction of $M$ in Conjecture~\ref{conj:sbo} is best possible, that is, a $k$-colorable matroid is not necessarily reducible to a $k$-colorable strongly base orderable matroid. 

Consider the graphic matroid $M$ of the complete graph $K_4$ on four vertices. It is not difficult to check that $M$ is not strongly base orderable. Suppose indirectly that $M$ is reducible to a 2-colorable strongly base orderable matroid $N$. Let $\cB_M$ and $\cB_N$ denote the sets of bases of $M$ and $N$, respectively. Then $\cB_M$ consists of 12 paths of length three and 4 stars. As $N$ is 2-colorable, its ground set can be partitioned into two disjoint bases, thus $\cB_N$ contains a path of length three and its complement (since the complement of a star with three edges is a triangle).

In what follows, we see $\cB_N = \cB_M$, which contradicts that $N$ is strongly base orderable.
Let $v_1, v_2, v_3, v_4$ denote the vertices of $K_4$. By the above, we may assume that $B_1 = \{v_1v_2, v_2v_3, \allowbreak v_3v_4\} \in \cB_N$ and $B_2 = \{v_2v_4, v_4v_1, v_1v_3\} \in \cB_N$.
By the symmetric exchange property, there exists an edge $e \in B_2$ such that $B_1-v_1v_2+e \in \cB_N$ and $B_2-e+v_1v_2 \in \cB_N$.
As $B_1-v_1v_2+v_2v_4=\{v_2v_3, v_2v_4, v_3v_4\} \not \in \cB_M$, $B_2-v_1v_3+v_1v_2 = \{v_1v_2, v_1v_4, v_2v_4\} \not \in \cB_M$ and $\cB_N \subseteq \cB_M$ hold, we get $e \ne v_2v_4$ and $e\ne v_1v_3$, hence $e = v_4v_1$. Therefore $B_1-v_1v_2+v_1v_4 = \{v_2v_3, v_3v_4, v_4v_1\} \in \cB_N$ and $B_2-v_4v_1+v_1v_3 = \{v_3v_1, v_1v_2, v_2v_4\} \in \cB_N$.
Using the same argument for these two bases, and then once again for the two bases thus obtained, we get that
\[\{v_iv_{i+1}, v_{i+1}v_{i+2}, v_{i+2}v_{i+3}\}, \{v_{i+1}v_{i+3}, v_{i+3}v_i, v_iv_{i+2}\} \in \cB_N\]
holds for $i=1,2,3,4$ (where indices are meant in a cyclic order). 

Applying the symmetric exchange property to $B_3 = \{v_iv_{i+1}, v_{i+1}v_{i+2}, v_{i+2}v_{i+3}\} \in \cB_N$, $B_4 = \{v_{i+2}v_i, v_iv_{i+1}, v_{i+1}v_{i+3}\} \in \cB_N$ and $v_{i+2}v_{i+3} \in B_3-B_4$, we get that there exists an $e \in B_4-B_3$ such that $B_3-v_{i+2}v_{i+3}+e \in \cB_N$ and $B_4-e+v_{i+2}v_{i+3}  \in \cB_N$.
As $B_3-v_{i+2}v_{i+3}+v_iv_{i+2}=\{v_iv_{i+1},v_iv_{i+2},v_{i+1}v_{i+2}\}\not \in \cB_N$, we get $e \ne v_iv_{i+2}$, hence $e=v_{i+1}v_{i+3}$. Therefore \[\{v_iv_{i+1}, v_{i+1}v_{i+2}, v_{i+1}v_{i+3}\}, \{v_{i+1}v_i, v_iv_{i+2}, v_{i+2}v_{i+3}\} \in \cB_N\]
for $i=1,2,3,4$.

It is not difficult to check that the we listed all 16 bases of $M$, therefore $\cB_N = \cB_M$ and $N=M$.
\end{rem}

\section*{Acknowledgements}

Krist\'of B\'erczi was supported by the J\'anos Bolyai Research Fellowship of the Hungarian Academy of Sciences and by the ÚNKP-19-4 New National Excellence Program of the Ministry for Innovation and Technology. Tam\'as Schwarcz was supported by the European Union, co-financed by the European Social Fund (EFOP-3.6.3-VEKOP-16-2017-00002). Yutaro Yamaguchi was supported by Overseas Research Program in Graduate School of Information Science and Technology, Osaka University. Projects no. NKFI-128673 and no. ED\_18-1-2019-0030 (Application-specific highly reliable IT solutions) have been implemented with the support provided from the National Research, Development and Innovation Fund of Hungary, financed under the FK\_18 and the Thematic Excellence Programme funding schemes, respectively. This  work  was  supported  by  the  Research  Institute for  Mathematical  Sciences,  an  International Joint Usage/Research Center located in Kyoto University.

\bibliographystyle{abbrv}
\bibliography{matroids}

\end{document}